\documentclass{eptcs}
\usepackage[utf8]{inputenc}
\usepackage{amsmath}
\usepackage{amssymb}
\usepackage{semantic}
\usepackage{graphicx} 
\usepackage{mathpartir}
\usepackage{hyperref}
\providecommand{\event}{ICLP 2024}

\newcounter{thm}
\newcounter{lmm}
\newcounter{coro}
\newcounter{propo}
\newcounter{defin}
\newcounter{exmpl}
\newcounter{princ}

\newenvironment{theorem}[1][]{\par\vspace{\baselineskip}\noindent\refstepcounter{thm}\textbf{Theorem \thethm #1: }}{\vspace{\baselineskip}}

\newenvironment{lemma}[1][]{\par\vspace{\baselineskip}\noindent\refstepcounter{lmm}\textbf{Lemma \thelmm #1: }}{\vspace{\baselineskip}}

\newenvironment{definition}[1][]{\par\vspace{\baselineskip}\noindent\refstepcounter{defin}\textbf{Definition \thedefin #1: }}{\vspace{\baselineskip}}
\newenvironment{example}{\par\vspace{\baselineskip}\noindent\refstepcounter{exmpl}\textbf{Example \theexmpl: }}{\vspace{\baselineskip}}

\newenvironment{proof}{\paragraph{Proof:}}{\hfill$\square$}

\title{Regular Typed Unification}

\author{Jo\~ao Barbosa \quad\quad Mário Florido
\institute{Departamento de Ci\^encia de Computadores, Faculdade de Ci\^encias, Universidade do Porto\\
Rua do Campo Alegre s/n, 4169--007 Porto, Portugal}
\institute{LIACC - Laborat\'orio de Intelig\^encia Artificial e Ci\^encia de Computadores}
\and
Vítor Santos Costa
\institute{Departamento de Ci\^encia de Computadores, Faculdade de Ci\^encias, Universidade do Porto\\
Rua do Campo Alegre s/n, 4169--007 Porto, Portugal}
\institute{INESCTEC - Instituto de Engenharia de Sistemas e Computadores, Tecnologia e Ciência}}

\def\titlerunning{Regular Typed Unification}
\def\authorrunning{J. Barbosa, M. Florido \& V. Santos Costa}

\begin{document}

\maketitle
\begin{abstract} 
Here we define a new unification algorithm for terms interpreted in semantic domains denoted by a subclass of regular types here called deterministic regular types. This reflects our intention not to handle the semantic universe as a homogeneous collection of values, but instead, to partition it in a way that is similar to data types in programming languages. 
We first define the new unification algorithm which is based on constraint generation and constraint solving, and then prove its main properties: termination, soundness, and completeness with respect to the semantics. Finally, we discuss how to apply this algorithm to a dynamically typed version of Prolog.
\end{abstract}

\section{Introduction}

In mathematical logic, a {\em term} denotes a mathematical object, and a
theory of equality on the set of all terms formally defines which
terms are considered equal. In logic programming terms are a syntactic
representations of structured data such that in the typical case of
first order languages, only syntactically identical terms are
considered equal. Functions are thus uninterpreted or,
computationally, functions build data terms, rather than operating on
them.

When the domain of discourse contains elements of different kinds, it is useful to split the set of all terms (Universe) accordingly. 
To this end, a {\em type} (sometimes also called {\em sort}) is assigned to variables and constant symbols, and a declaration of the domain type and range type to each function symbol. A typed term $f(t_1,...,t_n)$ may then be composed from the subterms $t_1,...,t_n$ only if the i-th subterm's type matches the declared i-th domain type of $f$. Such a term is called {\em well-typed} and terms which are not well-typed are called {\em ill-typed}.

Previous approaches for types in logic programming use {\em regular types} as the type language. Some examples of that work are the works by Zobel \cite{DBLP:conf/iclp/Zobel87,DBLP:books/mit/pfenning92/DartZ92}, Mishra \cite{Mishra84}, Yardeni \cite{YARDENI1991}, Fruhwirth \textit{et al}. \cite{FruhwirthSVY91}, Codish \cite{CodishL00}, Schrijvers \textit{et al}. \cite{schrijvers2008towards}, Bruynooghe \cite{SchrijversBG08}, Gallagher \cite{gallagher2004abstract}, Hermenegildo \textit{et al.} \cite{Hermenegildo00C23}, and Barbosa \textit{et al.} \cite{BarbosaFloridoCosta21}, among others.  

Data type definitions in programming languages impose constraints to the type language to allow decidable type checking. Namely, data types are recursive definitions where constructors are unique. Here we use such a subset of regular types that we shall call {\em deterministic regular types}, where each type constructor (here called {\em type function symbol}) is unique.

Type checking is most often done at compile-time, in order to ensure that  program execution will not generate type errors at run-time. If typing is not possible, type errors may occur as the values of one or more arguments may not be in the expected domain.
This work stems from the observation that logic programming systems will indeed output type errors in arguments of primitive predicates, but that there is no way to check if unification of terms from different domains occurs. As a result, a program may generate several (unreported) errors and still succeed. The user may receive an unexpected answer, while having no insight on the existence of actual execution errors, thus making it difficult to detect and resolve program bugs.
With this motivation in mind, we design a typed unification algorithm for
typed first order theories, where types are described by deterministic regular types.
This new unification algorithm may return three different results: a {\em most general unifier}, {\em failure} or {\em wrong}. This last value {\em wrong} is inspired by a similar notion used by Robin Milner to denote run-time type errors in functional programs \cite{DBLP:journals/jcss/Milner78} and, in our framework, it corresponds to the unification of terms that can never belong to the same semantic domains.
A function now, may map integers to integers, integers to lists, floats to lists of integers, and, thus, the Herbrand universe is now divided into many different domains.

\begin{example}
Let $cons/2$ be the list constructor (in Prolog it would be denoted by $./2$) with type $\forall \alpha. \alpha \times list(\alpha) \to list(\alpha)$, where $list(\alpha) = [~] + [\alpha~|~list(\alpha)]$ ($+$ denotes type union). If we have terms $t_1 = cons(1,X)$ and $t_2 = cons(Y,2)$. These terms unify using first order (untyped) unification, but do not have a correct type, since the second argument of the list constructor must be a list. This is captured by the typed unification algorithm since it outputs {\em wrong}.
\end{example}

{\bf Contributions} Our main contributions are: an extension to the semantics defined in \cite{BarbosaFloridoCosta22}, where equality takes into account the domains of the two terms in the left and the right hand side of an equation, being {\em wrong} when terms belong to disjoint domains; a {\em type system} for terms and the equality predicate which we prove to be sound with respect to the semantic typing relation; and a new unification algorithm, which given an equation between two terms returns a most general unifier for them and a {\em principal type}, if there is a solution, {\em false} if there is no solution but terms can belong to the same semantic domain and {\em wrong} otherwise. 
This three stage framework (first a notion of semantic typing, then a type system for terms and equations sound with respect to semantic typing and, finally, a unification algorithm sound and complete with respect to the type system) 
enables us to smoothly prove soundness and completeness of our unification algorithm, and it is inspired by the type theory in \cite{DBLP:journals/jcss/Milner78}. 

\section{Term Syntax and Semantics}\label{S2}

Here we define the language of terms, following \cite{Apt:1996:LPP:249573,Lloyd:1984:FLP:2214}.
Given an infinite set of variables \textbf{VAR} and an infinite set of function symbols \textbf{FUNC}, a term is:
\begin{enumerate}
    \item a variable (\texttt{X}, \texttt{Y}, \texttt{Xi}, \dots);
    \item a function symbol of arity 0 (\texttt{k}, \texttt{a}, \texttt{b}, \texttt{1}, \dots), which we call a constant;
    \item a function symbol of arity $n\geq1$ (\texttt{f}, \texttt{g}, \texttt{h}, \dots) applied to an n-tuple of terms.
\end{enumerate}

We call terms that contain no variables \textit{ground terms}, and terms that start with a function symbol with arity $n\geq1$ \textit{complex terms}.

Following the standard Herbrand interpretation of logic programs \cite{Apt:1996:LPP:249573,Lloyd:1984:FLP:2214}, we assume that every ground term represents a tree and that all these trees are part of the universe of interpretation of the logic program.

We assume a particular partition of the universe into several domains. This interpretation groups sets of trees in the universe into domains, and includes some other domains that are not consisting of trees. We divide the universe \textbf{U} into domains as follows:
$$\mathbf{U} = \mathbf{Int} \cup \mathbf{Flt} \cup \mathbf{Str} \cup \mathbf{Atm} \cup \mathbf{List_1} \cup \dots \cup \mathbf{List_n} \cup \mathbf{A_1} \cup \dots \cup \mathbf{A_{m}} \cup \mathbf{Bool} \cup \mathbf{F} \cup \mathbf{Wrong},$$

\noindent where \textbf{Int} is the set of trees that represent integers (examples include \texttt{1} and \texttt{-10}, but also trees such as \texttt{1 + 4} and \texttt{2 * 5 - 1}), \textbf{Flt} is the set of trees that represent floating-point numbers, \textbf{Str} is the set of trees representing strings, \textbf{Atm} is the set of trees consisting of a single node, the root, that are not included in any other domain, $\mathbf{List_i}$ are sets of trees that represent lists, where each domain contains the trees that represent lists of elements of some other domain (i.e., we have a domain for lists of integers, lists of strings, lists of lists of integers, \dots), $\mathbf{A_i}$ are the domains of trees whose root is a function symbol and the nodes of each tree are in the same domain as the corresponding nodes of every other tree (examples: $f(Int)$, $g(Int, Float)$, $h(g(Atom),h(Int))$, \dots), \textbf{Bool} is the set with \textit{true} and \textit{false}, \textbf{F} is the set of functions, and \textbf{Wrong} is the set with the single value, \textit{wrong}. We call {\em base domains} the domains \textbf{Int}, \textbf{Flt}, \textbf{Str}, and \textbf{Atm}.

One important note here is the value \texttt{[]}, corresponding to the empty list. We assume that this value belongs to every list domain, and that it is the only value that belongs to more than one domain in this partition.

We are using the domains for lists as an example of an interesting division of the universe that will later on correspond to inductively defined types. We could easily extend this partition by adding domains for other data types such as binary trees. We believe that any further domains for structured data can be extended easily following the approach we have for lists.

The semantics of a term is a tree in some domain, or \textit{wrong}. The semantics depends on an interpretation $I$ for the function symbols in the language, and a state $\Sigma$ which associates variables to semantic values. We assume that the value returned by $I$ is, for constants, a tree with just a root, and for function symbols of arity $n\geq1$ a function in \textbf{F} that outputs a tree. Without loss of generality we assume that the only function symbol which does not have an Herbrand interpretation is the list constructor, thus for all function symbols \texttt{f} except the list constructor \texttt{cons}, the corresponding function in $I$ is a function \textit{f} that has signature \textit{f} $: \forall \alpha_1,\dots,\alpha_n. \alpha_1 \times \dots \times \alpha_n \to f(\alpha_1,\dots,\alpha_n)$, such that if any of the arguments the function is applied to is $wrong$ then it outputs $wrong$, otherwise it outputs the tree with root $f$ and children the trees it got as input. For the list constructor \texttt{cons} the function associated in $I$ is \textit{cons} with signature \textit{cons} $: \forall \alpha. \alpha \times list(\alpha) \to list(\alpha)$ defined as:
\begin{equation*}
    \textit{cons}(v_1,v_2) =
    \begin{cases}
      cons(v_1,v_2) & \text{if }v_1 \in D \wedge v_2 \in List(D)\\
      wrong & \text{otherwise}  \\
    \end{cases}
 \end{equation*}

The predefined interpretation $I$ is the one where every constant has the expected value, for instance the term \texttt{1} has as value the integer $1$. One additional useful definition is the function $dom$ that returns a set of domain of a value, a singleton set for all values but $[~]$.

We define the semantics of a term, represented by $|[~|]_{I,\Sigma}$, in the following way:
\begin{itemize}
    \item $|[X|]_{I,\Sigma} = \Sigma(X)$
    \item $|[k|]_{I,\Sigma} = I(k)$
    \item $|[f(t_1,\dots,t_n)|]_{I,\Sigma} = I(f) (|[t_1|]_{I,\Sigma},\dots,|[t_n|]_{I,\Sigma})$
\end{itemize}

Note that, if a complex term contains the list constructor, the semantics of that term can be $wrong$. This is where the division into domains comes into play, since if we were considering an undivided Herbrand universe, then trivially all values are in the same domain so the application of a function could never generate an error. Informally, our approach supports the division of the universe implicit on the abstract data types definitions in the Prolog ISO standard (4.2)~\cite{deranseratbook96}.



We assume only one predicate, equality, hereby represented by $=$. The semantics of equality is re-adjusted to take into account the value $wrong$. Equality is defined for terms in the same domain. So let function \textit{eq} with signature \textit{eq} $:\forall \alpha. \alpha \times \alpha \to Bool$ be defined as follows:
\begin{equation*}
    \textit{eq}(v_1,v_2) =
    \begin{cases}
      true & \text{if }v_1 = v_2 \wedge dom(v_1) \cap dom(v_2) \neq \emptyset \wedge dom(v_1) \neq \{Wrong\}\\
      false & \text{if }v_1 \neq v_2 \wedge dom(v_1) \cap dom(v_2) \neq \emptyset \wedge dom(v_1) \neq \{Wrong\} \\
      wrong & \text{otherwise} \\
    \end{cases}
 \end{equation*}

\noindent
The semantics for the equality predicate is then as follows:
    $$|[t_1 = t_2|]_{I,\Sigma} = eq(|[t_1|]_{I,\Sigma},|[t_2|]_{I,\Sigma})$$


\section{Types}

Types are syntactic descriptions of semantic domains. The alphabet for the language of types includes an infinite set of type variables \textbf{TVar}, a finite set of base types \textbf{TBase}, an infinite set of type function symbols \textbf{TFunc}, an infinite set of type symbols \textbf{TSym}, parenthesis, and the comma. There is a one-to-one correspondence between \textbf{TFunc} and \textbf{FUNC}, which we assume is predefined.
Then, we have the following grammar for types:\\

\begin{tabular}{l c l}
$all\_type$  & ::= & $cons\_type~|~func\_type$\\
$cons\_type$ & ::= & $type~|~type\_term~|~bool$\\
$func\_type$ & ::= & $type_1 \times \dots \times type_n \to type~|~type_1 \times \dots \times type_n \to bool$\\
$type$       & ::= & $tvar~|~tbase~|~tsymbol(type_1,\dots,type_n)$\\
$type\_term$ & ::= & $tconstant~|~tfunction(cons\_type_1,\dots,cons\_type_n)$\\
$type\_def$  & ::= & $tsymbol(tvar_1,\dots,tvar_n) --> type\_term_1 + \dots + type\_term_m$
\end{tabular}\\

\noindent
where $tvar \in$ \textbf{TVar}, $tbase \in$ \textbf{TBase}, $tconstant$ and $tfunction \in$ \textbf{TFunc}, and $tsymbol \in$ \textbf{TSym}. We call a type term that starts with a $tfunction$ a complex type term. We call \emph{ground} to any type that does not contain a type variable. 

Each type symbol is defined by a {\em type definition}. A {\em well-formed} type definition has all type variables that occur as parameters on the left-hand side of the definition be distinct and occurring somewhere on the right-hand side, and all type variables that occur on the right-hand side be a parameter on the left-hand side. The sum $\tau_1 + \dots + \tau_n$ is a {\em union type}, describing values that have one of the type terms $\tau_1,\ldots,\tau_n$, called the {\em summands}. The `+' is an idempotent, commutative, and associative operation.

A set of type definitions $D$ is called {\em deterministic} if it is well-formed and any type function symbol occurs at most once in $D$. In \cite{DBLP:books/mit/pfenning92/DartZ92}, the authors introduce the concept of \textit{deterministic type definition}. Our definition is stricter than this previous one by disallowing base types and variables as summands in type definitions and disallowing more than one occurrence of any function symbol in the whole set of type definitions.

Deterministic type definitions include tuple-distributive types \cite{DBLP:conf/iclp/Zobel87, Mishra84} and correspond to the widely used algebraic data types in programming languages. From now on we assume that type definitions are deterministic.

A \emph{type scheme} $\sigma$ is an expression of the form $\forall \alpha_1,\dots,\alpha_n.\tau$, where $\tau$ is a $type$ or a $func\_type$ and $\alpha_1,\dots,\alpha_n$ are type variables which will be called the {\em generic variables} of $\sigma$. If $\tau$ has no variables, then it is itself a type scheme. Note that types form a subclass of type schemes. We will abbreviate type schemes to $\forall \vec{\alpha}.\tau$, where $\vec{\alpha}$ denotes a sequence of several type variables $\alpha_i$. Type
 schemes represent parametric polymorphic types \cite{DBLP:conf/popl/DamasM82}.

\section{Semantics}
Each instance of a $type$ is associated with a domain. A base type is associated with a base domain, and each instance of a type of the form $tsymbol(type_1,\dots,type_n)$ is associated with a domain. We include a type symbol $list$ that is associated with the domains for lists. We assume that the definition for the type symbol $list$ is: $list(\alpha) --> [~] + cons(\alpha,list(\alpha))$. We could include further type symbols that were defined by inductive definitions, besides lists, and the rest of this paper could be easily extended to include different inductively defined types, but we keep {\em list} as the only one for the sake of simplicity.

A valuation $\psi$ maps each type variable to a ground type. Given a valuation $\psi$, we define the semantics of a $type$ and the type $bool$ as follows:\\

\noindent
$\mathbf{T}|[\alpha|]_{\psi} =$ $\mathbf{T}|[\psi(\alpha)|]_{\psi}$\\
$\mathbf{T}|[int|]_{\psi} =$ \textbf{Int}\\
$\mathbf{T}|[bool|]_{\psi} =$ \textbf{Bool}\\
$\mathbf{T}|[float|]_{\psi} =$ \textbf{Flt}\\
$\mathbf{T}|[string|]_{\psi} =$ \textbf{Str}\\
$\mathbf{T}|[atom|]_{\psi} =$ \textbf{Atm}\\
$\mathbf{T}|[list(\alpha)|]_{\psi}$ = $\mathbf{T}|[list(\psi(\alpha))|]_{\psi}$\\
$\mathbf{T}|[list(int)|]_{\psi} = List_i$, where $List_i$ is the domain for lists of integer. Similarly for any other ground instance of $list(\alpha)$ and the corresponding domain $List_j$.\\

The semantics of a $type\_term$ is:\\
$\mathbf{T}|[k|]_{\psi} = \{k\}$\\
$\mathbf{T}|[f(\tau_1,\dots,\tau_n)|]_{\psi}$ = $\{ f(v_1,\dots,v_n)~|~v_i \in \mathbf{T}|[\tau_i|]_{\psi}\}$\\

The semantics of a \textit{union type} is:\\
$\mathbf{T}|[\tau_1 + \dots + \tau_n|]_{\psi} = \mathbf{T}|[\tau_1|]_{\psi} \cup \dots \cup \mathbf{T}|[\tau_n|]_{\psi}$\\

The semantics of a $func\_type$ is:\\
$\mathbf{T}|[\tau_1 \times \dots \times \tau_n \to \tau|]_{\psi} = 
\{ f~|~f \in \mathbf{F} \wedge (v_1 \in \mathbf{T}|[\tau_1|]_{\psi} \wedge \dots \wedge v_n \in \mathbf{T}|[\tau_n|]_{\psi} \implies f(v_1,\dots,v_n) \in \mathbf{T}|[\tau|]_{\psi}) \}$\\

The semantics of a \emph{type scheme} is:\\
$\mathbf{T}|[\forall \vec{\alpha}.\tau|]_{\psi} = \bigcap_{\forall\vec{\sigma}} \mathbf{T}|[\tau[\vec{\alpha}\mapsto\vec{\sigma}]|]_{\psi}$, where $\vec{\sigma}$ is a sequence of $types$ of the same size as $\vec{\alpha}$.\\

Note that the semantics of a (ground) $type\_term$ may be a domain, as in the case of $f(int,float)$, or the subset of a domain, as in the case of $cons(int,[~])$, or even a subset of several domains, as in the case of $[~]$. This includes the domain \textbf{Wrong}, as in the case of $cons(int,int)$. All instances of complex type terms whose $tfunction$ is not the list constructor are associated with a domain for trees.

Also note that we assume a function, given by $I$, for the interpretation of function symbols, thus functions have type signatures: the type of a function symbol $f$ of arity $n$ is interpreted as a function which builds a tree of root $f$, with the type scheme $\forall \alpha_1,\dots,\alpha_n. \alpha_1 \times \dots \times \alpha_n \to f(\alpha_1,\dots,\alpha_n)$. The semantics for this type scheme is the intersection of the semantics for all instances of the functional type, which is a subset of \textbf{F} consisting of all functions that have such type. So it consists of all the functions that can have any tuple of n elements as input and output a tree whose root is $f$ and the children nodes are the input elements.

\subsection{Semantic Typing}

We now define what it means for a term to semantically have a type, denoted by $t : \tau$. If the term and the type are both ground, given an interpretation $I$, we just check whether the semantics of the term belongs to the domain corresponding to the semantics of the type. So, for ground terms and types:
$$t : \tau \implies \forall \Sigma.\forall \psi. |[t|]_{I,\Sigma} \in \mathbf{T}|[\tau|]_{\psi}$$

However, both terms and types can be non-ground in general and, without extra information, we cannot know what is the correct type for a variable. To deal with variables we introduce the concept of a context $\Gamma$, defined as a set of typings of the form $X : \tau$ for variables. Given a context we define the {\em semantic typing} relation, denoted by $\models$, as:
$$\Gamma \models_{I} t:\tau \implies \forall \Sigma.\forall \psi. (\forall (X:\tau\prime) \in \Gamma. |[X|]_{I,\Sigma} \in \mathbf{T}|[\tau\prime|]_{\psi} \implies  |[t|]_{I,\Sigma} \in \mathbf{T}|[\tau|]_{\psi})$$

We call the \textit{generic context} to the context that contains $X_i : \alpha_i$, for all variables, i.e., all term variables are typed by a type variable, and each type variable is associated with a particular variable. Note that throughout the paper we will use the symbol $\models$ overloaded for other semantics relations.

\begin{example}
Let $\Gamma = \{\texttt{X}:\alpha,\texttt{Y} : list(\alpha)\}$ and the type signature for \texttt{cons} in $I$ be $\{\texttt{cons} : \forall \alpha. \alpha \times list(\alpha) \to list(\alpha)\}$.
$$\Gamma \models_{I} \texttt{cons(X,Y)} : list(\alpha)$$
Suppose we have a state $\Sigma$ and a valuation $\psi$ such that $|[\texttt{X}|]_{I,\Sigma} \in \mathbf{T}|[\alpha|]_{\psi}$ and $|[\texttt{Y}|]_{I,\Sigma} \in \mathbf{T}|[list(\alpha)|]_{\psi}$, then $|[ \texttt{cons(X,Y)} |]_{I,\Sigma} \in \mathbf{T}|[list(\alpha)|]_{\psi}$. Since $|[\texttt{X}|]_{I,\Sigma} = \Sigma(X) \in \mathbf{T}|[\psi(\alpha)|]_{\psi}$ and $|[\texttt{Y}|]_{I,\Sigma} = \Sigma(Y) \in \mathbf{T}|[list(\alpha)|]_{\psi} = \mathbf{T}|[list(\psi(\alpha))|]_{\psi}$, by the semantics of \texttt{cons}, we have $cons(\Sigma(X),\Sigma(Y))$, which is not $wrong$ from the domains of the respective values, and because the output is in the correct domain.\\
However, note that for $\Gamma\prime = \{\texttt{X}:\alpha,\texttt{Y} : \beta\}$, the same would not be true, since for $\Sigma = [X \mapsto 1, Y \mapsto 2]$ and $\psi = [\alpha \mapsto int, \beta \mapsto int]$, the left-hand side of the implication is true, but $cons(1,2) \notin \mathbf{T}|[list(int)|]_{\psi}$.
\end{example}

\section{Syntactic Typing}

Syntactic typing is defined by a type system.
A context $\Gamma$ and a set of type assumptions for constants and function symbols $\Delta$ are needed to derive a type assignment and one writes $\Gamma,\Delta \vdash t:\tau$ (pronounce this as $\Gamma$ and $\Delta$ yield $t$ in $\tau$).
Assumptions in $\Delta$ are of the form $k:\forall \vec{\alpha}.\tau$, for constants, and $f : \forall \vec{\alpha}.\tau_1\times \dots \times \tau_n \to \tau$, for function symbols, where the generic variables $\vec{\alpha}$ of these type schemes are exactly the type variables that occur in $\tau$ and $\tau_1 \times \dots \times \tau_n \to \tau$, respectively. 
A statement $t:\tau$ is {\em derivable} from contexts $\Gamma$ and $\Delta$, notation $\Gamma,\Delta \vdash t:\tau$, if it can be produced by the rules in Figure \ref{TSystem}. If we have a derivation in the type system, then we say that $t$ has type $\tau$ in contexts $\Gamma$ and $\Delta$.

\begin{figure*}[!ht]
\begin{mathpar}

\inferrule* [left=\textsc{VAR}]
{(X:\tau) \in \Gamma }
{\Gamma, \Delta \vdash X:\tau}

\inferrule* [left=\textsc{CST}]
{(k : \forall \vec{\alpha}.\tau) \in \Delta}
{\Gamma, \Delta \vdash k:\tau[\vec{\alpha} \mapsto \vec{\sigma}] }

\inferrule* [left=\textsc{CPL}]
{(f : \forall \vec{\alpha}.\tau_1 \times \dots \times \tau_n \to \tau) \in \Delta \\\\
\Gamma,\Delta \vdash t_1: \tau_1[\vec{\alpha} \mapsto \vec{\sigma}] ~~\dots ~~ \Gamma, \Delta \vdash t_n:\tau_n[\vec{\alpha} \mapsto \vec{\sigma}]}
{\Gamma, \Delta \vdash f(t_1,\dots,t_n) : \tau[\vec{\alpha} \mapsto \vec{\sigma}]}

\inferrule* [left=\textsc{EQU}]
{\Gamma, \Delta \vdash t_1 : \tau ~~~~~ \Gamma, \Delta \vdash t_2 : \tau}
{\Gamma, \Delta \vdash t_1 = t_2 : bool}
\end{mathpar}
\caption{Type System}
\label{TSystem}
\end{figure*}


We must guarantee that $\Delta$ is in agreement with $I$. For this, we have the following relation: $I \models \Delta$, is defined as $\forall (k : \tau) \in \Delta. dom(I(k)) = \{\tau\} \wedge \forall (f : \tau_1 \times \dots \times \tau_n \to \tau). I(f) : \tau_1 \times \dots \times \tau_n \to \tau$. 

\begin{example}
Let $\Gamma = \{X : int,~Y : list(int)\}$, $\Delta = \{1 : int,~nil : \forall \gamma. list(\gamma),~cons : \forall \beta. \beta \times list(\beta) \to list(\beta)\}$, and $\Lambda = (cons : \forall \beta. \beta \times list(\beta) \to list(\beta) ) \in \Delta$ (we use this $\Lambda$ just to improve presentation). Then the following type derivation holds using the type rules:

\begin{mathpar}
\inferrule*[]
{   \inferrule* []
    {
        \inferrule*[]
        {(X : int)\in \Gamma}
        {\Gamma, \Delta \vdash X : int}
        ~~~~
        \inferrule*[]
        {([~] : \forall \gamma. list(\gamma)) \in \Delta}
        {\Gamma, \Delta \vdash [~] : list(int)^{(2)}}
        ~~
        \Lambda
    }
    {\Gamma, \Delta \vdash cons(X,[~]): list(int)^{(1)}}
    ~~~~
    \inferrule* []
    {
        \inferrule*[]
        {(1 : int) \in \Delta}
        {\Gamma, \Delta \vdash 1 : int}
        ~~~~
        \inferrule*[]
        {(Y : list(int)) \in \Gamma}
        {\Gamma, \Delta \vdash Y : list(int)}
        ~~
        \Lambda
    }
    {\Gamma, \Delta \vdash cons(1,Y): list(int)^{(2)} }
}
{ \Gamma, \Delta \vdash cons(X,[~]) = cons(1,Y): bool }
\end{mathpar}

Note that in $^{(1)}$ we used $list(\beta)[\beta \mapsto int]$ and in $^{(2)}$ we used $list(\gamma)[\gamma \mapsto int]$. Also note that if $X : \alpha$ instead of $X:int$ was in $\Gamma$, we could not have a derivation.
\end{example}

We now prove that the rules for syntactic typing are sound, that is, if the set $\Delta$ is in agreement with $I$, then any type derivation is semantically correct.

\begin{theorem}[ - Soundness of Syntactic Typing]\label{SoundTS}
If $\Gamma, \Delta \vdash t:\tau$ and $I\models \Delta$, then $\Gamma \models_I t:\tau$.
\end{theorem}
\begin{proof}
We will prove this by induction on the derivation.
\begin{itemize}
    \item If the term $t$ is a variable $X$, then the derivation consists of a single application of axiom $VAR$. Clearly, it is also true that $\Gamma \vdash_I X:\tau$, where $(X:\tau) \in \Gamma$, since any $\Sigma$ that gives values to $X$ and $\psi$ that gives values to $\tau$, such that $|[X|]_{I,\Sigma} \in \mathbf{T}|[\tau|]_{\psi}$ will do so in the context and in the term itself simultaneously, so $\Gamma \vdash_I X : \tau$.

    \item If the term $t$ is a constant $k$, then the derivation consists of a single application of axiom $CST$. Since $I \models \Delta$, $dom(I(k)) = \{\forall \vec{\alpha}.\tau\}$, where $(k : \forall \vec{\alpha}.\tau) \in \Delta$, then $k \in \mathbf{T}|[\forall \vec{\alpha}.\tau|]_{\psi}$, for any $\psi$. But since $\mathbf{T}|[\forall \vec{\alpha}.\tau|]_{\psi} = \bigcap_{\forall\vec{\sigma}} \mathbf{T}|[\tau [\vec{\alpha} \mapsto \vec{\sigma}]|]_{\psi}$, then $k \in \mathbf{T}|[\tau [\vec{\alpha} \mapsto \vec{\sigma}]|]_{\psi}$. So for any $\Sigma$, the right-hand side of the implication is always true, so $\Gamma \vdash_I k : \tau [\vec{\alpha} \mapsto \vec{\sigma}]$.

    \item If the term $t$ is a complex term $f(t_1,\dots,t_n)$, then we can assume, by induction hypothesis, that $\Gamma \vdash_I t_i: \tau_i[\vec{\alpha} \mapsto \vec{\sigma}]$, for all $i = 1,\dots,n$. Since $I \models \Delta$, $I(f): \forall \vec{\alpha}. \tau_1 \times \dots \times \tau_n \to \tau$, then $f \in \mathbf{T}|[\forall \vec{\alpha}. \tau_1 \times \dots \times \tau_n \to \tau|]_{\psi}$, for all $\psi$, so $f \in \mathbf{T}|[(\tau_1 \times \dots \times \tau_n \to \tau)[\vec{\alpha} \mapsto \vec{\sigma}]|]_{\psi}$. Therefore, we know that, if $v_i \in \mathbf{T}|[\tau_i[\vec{\alpha} \mapsto \vec{\sigma}]|]_{\psi}$ then $f(v_1,\dots,v_n) \in \mathbf{T}|[\tau[\vec{\alpha} \mapsto \vec{\sigma}]|]_{\psi}$. For any $\Sigma$ and $\psi$ such that $\forall (X:\tau\prime) \in \Gamma. |[X|]_{I,\Sigma} \in \mathbf{T}|[\tau\prime|]_{\psi}$, 
    by the induction hypothesis $|[t_i|]_{I,\Sigma} \in \mathbf{T}|[\tau_i[\vec{\alpha} \mapsto \vec{\sigma}]|]_{\psi}$. Therefore, for the same $\Sigma$ and $\psi$, we know that $|[f(t_1,\dots,t_n)|]_{I,\Sigma} \in \mathbf{T}|[\tau[\vec{\alpha} \mapsto \vec{\sigma}]|]_{\psi}$, so $\Gamma \vdash_I f(t_1,\dots,t_n) : \tau[\vec{\alpha} \mapsto \vec{\sigma}]$.
    
    \item If we have an equality of two terms $t_1 = t_2$, we can assume, by induction hypothesis, that $\Gamma \vdash_I t_1: \tau$ and $\Gamma \vdash_I t_2: \tau$. Therefore we know that for any $\Sigma$ and $\psi$ such that $\forall (X:\tau\prime) \in \Gamma. |[X|]_{I,\Sigma} \in \mathbf{T}|[\tau\prime|]_{\psi}$, we have $|[t_1|]_{I,\Sigma} \in \mathbf{T}|[\tau|]_{\psi}$ and $|[t_2|]_{I,\Sigma} \in \mathbf{T}|[\tau|]_{\psi}$. So for these $\Sigma$ and $\psi$, we have $|[t_1 = t_2|]_{\Sigma} \in |[bool|]_{\psi}$. Therefore, $\Gamma \vdash_I t_1 = t_2 : bool$.
\end{itemize}
\end{proof}

Given a term $t$ is there a typing representing all possible typings of $t$? In order to answer this question we introduce  the notion of {\em principal typing}, \cite{TrevorJim96},  as appropriate to our system.

\begin{definition}
A \emph{principal typing} is a pair $(\Gamma,\tau)$, such that $\Gamma, \Delta \vdash t:\tau$ and for every other pair $(\Gamma\prime,\tau\prime)$ such that $\Gamma\prime, \Delta \vdash t : \tau\prime$, there is a type substitution $\mu$ such that $\mu(\Gamma) = \Gamma\prime$ and $\mu(\tau) = \tau\prime$.
\end{definition}

Note that even though, initially, it might seem possible that the context in a principal typing will always be a generic context, for some cases that is not the case.

\begin{example}
Let $t =$\texttt{cons(X,Y)}. A principal typing for $t$ is $(\{X : \alpha, Y : list(\alpha)\}, list(\alpha))$. Note that any renaming of type variable $\alpha$ defines another principal typing, because principal typings are unique up to renaming of type variables. Also note that the type for $Y$ cannot be a type variable, thus, in this example, the context is not generic.
\end{example}

\section{Constraints}

To check implicit types during unification, we must deduce types that are not present in equality equations. To represent this problem in a broader context, we introduce the notion of type constraint which we add to the usual term unification problem.

We define equality constraints between terms $t_1 = t_2$, and equality constraints between types $\tau_1 \doteq \tau_2$. We are here using the same symbol for equality constraints and the equality predicate. We argue that the uses are clear from the context.

We say that a set of equality constraints is in {\em normal form} if all constraints are of the form $X_i = t_i$, for some term $t_i$, and there is no other occurrence of any $X_i$ anywhere else in the set. A set of equality constraints in normal form can be interpreted as a substitution, where every constraint of the form $X_i = t_i$ is interpreted as $[X_i \mapsto t_i]$.

A set of type equality constraints is in normal form if all constraints are of the form $\alpha_i \doteq \tau_i$, for some type $\tau_i$, and there is no other occurrence of any $\alpha_i$ anywhere else in the set. A set of type equality constraints in normal form can be interpreted as a type substitution, where every constraint of the form $\alpha_i \doteq \tau_i$ is interpreted as $[\alpha_i \mapsto \tau_i]$.

\begin{definition}
A substitution $\theta$ (or type substitution $\mu$) is called a {\em unifier} for terms $t_1$ and $t_2$ (or types $\tau_1$ and $\tau_2$), iff $\theta(t_1) = \theta(t_2)$ (or $\mu(\tau_1) = \mu(\tau_2)$). Terms $t_1$ and $t_2$ (or types $\tau_1$ and $\tau_2$) are {\em unifiable} iff there exists a unifier for them.
\end{definition}

Our constraints are supposed to represent equality, either of terms or types. However, in the semantics, we need states and valuations to interpret non-ground terms and types, respectively. Therefore, we need a way to interpret constraints semantically, so we define the following.

\begin{definition}
Let $c$ be a constraint, $\Sigma$ a state, and $\psi$ a valuation. We say that $\Sigma$ and $\psi$ model $c$, and represent it by $\Sigma,\psi \models c$ if:
\begin{itemize}
    \item $c$ is an equality constraint of the form $t_1 = t_2$, then $|[t_1|]_{I,\Sigma} = |[t_2|]_{I,\Sigma}$;
    \item $c$ is a type equality constraint of the form $\tau_1 \doteq \tau_2$, then $\mathbf{T}|[\tau_1|]_{\psi} = \mathbf{T}|[\tau_2|]_{\psi}$;
\end{itemize}
\end{definition}

We can easily extend this definition for sets of constraints.

\begin{definition}
Let $C$ be a set of equality constraints and $S$ be a set of type equality constraints. We say that a state $\Sigma$ and a valuation $\psi$ model the pair $(C,T)$, and represent it by $\Sigma,\psi \models C,T$ iff $\Sigma$ and $\psi$ model all constraints in both sets.
\end{definition}

We now provide an auxiliary definition that relates substitutions and states and use this definition to extend our notion of constraint modelling.

\begin{definition}
We say that a state $\Sigma$ follows a substitution $\theta$ and represent it by $\Sigma \sim \theta$ iff for any term $t$, $|[t|]_{I,\Sigma} = v$ and $|[\theta(t)|]_{I,\Sigma} = v$. Similarly, a valuation $\psi$ follows a substitution for types $\mu$ ($\psi \sim \mu$) iff for any type $\tau$, $\mathbf{T}|[\tau|]_{\psi} = \mathbf{T}|[\mu(\tau)|]_{\psi}$.
\end{definition}

\begin{definition}
Let $C$ be a set of equality constraints and $S$ be a set of type equality constraints. We say that a substitution $\theta$ and a type substitution $\mu$ model the pair $(C,T)$, and represent it by $\theta,\mu \models C,T$, iff for every state $\Sigma$ and valuation $\psi$ we have that $\Sigma \sim \theta \wedge \psi \sim \mu \implies \Sigma,\psi \models C,T $.
\end{definition}

\section{Typed Unification Algorithm}

The typed unification algorithm performs unification for terms and types. The intuition is that if the types do not unify, then there is a type error. We will prove this condition in the next section.
We follow the approach of \cite{Wand87}: generate constraints for typeability and solve them. 

\subsection{Constraint Generation}

\small
\begin{figure*}[!ht]
\begin{mathpar}
\inferrule* [left=\textsc{GVAR}]
{ (X: \alpha) \in \Gamma }
{\Gamma,\Delta \vdash X:\alpha ~|~ \emptyset~|~\emptyset}

\inferrule* [left=\textsc{GCST}]
{(k : \forall \vec{\alpha}.\tau) \in \Delta}
{\Gamma,\Delta \vdash k:\tau[\vec{\alpha} \mapsto \vec{\beta}]  ~|~ \emptyset ~|~ \emptyset}

\inferrule* [left=\textsc{GCPL}]
{(f : \forall \vec{\alpha}.\tau_1 \times \dots \times \tau_n \to \tau) \in \Delta \\ \\
\Gamma,\Delta \vdash t_1: \tau_1\prime ~|~ \emptyset ~|~ T_1~~\dots ~~ \Gamma, \Delta \vdash t_n:\tau_n\prime~|~ \emptyset ~|~ T_n}
{\Gamma, \Delta \vdash f(t_1,\dots,t_n) : \tau[\vec{\alpha} \mapsto \vec{\beta}]~|~\emptyset ~|~ T_1 \cup \dots \cup T_n  \cup \{\tau\prime_1 \doteq \tau_1[\vec{\alpha} \mapsto \vec{\beta}], \dots, \tau\prime_n \doteq \tau_n[\vec{\alpha} \mapsto \vec{\beta}] \}}

\inferrule* [left=\textsc{GEQU}]
{\Gamma,\Delta \vdash t_1 : \tau_1 ~|~C_1~|~T_1 ~~~~~ \Gamma,\Delta \vdash t_2 : \tau_2 ~|~C_2~|~T_2}
{\Gamma,\Delta \vdash t_1 = t_2 : bool ~|~\{t_1 = t_2\}~|~ T_1 \cup T_2 \cup \{\tau_1 \doteq \tau_2\}}
\end{mathpar}
\caption{Constraint Typing Judgment}
\label{ConsGen}
\end{figure*}
\normalsize

Guided by the definition of our type system we now define a {\em constraint typing judgment}, which indicates what constraints must hold for a particular type term-and-context pair to be typeable.

Let $\Gamma$ be a generic context, and $\Delta$ a set of type assumptions
for constants and function symbols. We use the following rules to generate constraints for the unification of two terms $t_1$ and $t_2$. The generated constraints will be the pair $(C,T)$ in $\Gamma, \Delta \vdash t_1 = t_2 : bool ~|~C~|~T$. In the rules in Figure \ref{ConsGen}, $\vec{\beta}$ represents a sequence of fresh type variables of the same size as $\vec{\alpha}$ in the corresponding case.

\begin{example}
Let $\Gamma$ be a generic context (we will denote the type variable associated with each variable $X$ by $\alpha_X$), $\Delta = \{1 : int,~[~] : \forall \alpha. list(\alpha),~cons : \forall \beta. \beta \times list(\beta) \to list(\beta)\}$, $C = \{cons(X,[~]) = cons(1,Y)\}$, and $\Lambda = (cons : \forall \beta. \beta \times list(\beta) \to list(\beta)) \in \Delta$. The following constraint type judgements hold:

\begin{mathpar}
\inferrule*[]
{
    \inferrule*[]
    {(X : \alpha_X)\in \Gamma}
    {\Gamma, \Delta \vdash X : \alpha_X~|~\emptyset~|~\emptyset}
    ~~~~
    \inferrule*[]
    {([~] : \forall \alpha. list(\alpha)) \in \Delta}
    {\Gamma, \Delta \vdash [~] : list(\gamma)~|~\emptyset~|~\emptyset}
    ~~
    \Lambda
}
{\Gamma, \Delta \vdash cons(X,[~]): list(\nu)~|~\emptyset~|~\{\alpha_X = \nu,~list(\gamma) = list(\nu)\} (= T_1)}
\\
\inferrule* []
{
    \inferrule*[]
    {(1 : int) \in \Delta}
    {\Gamma, \Delta \vdash 1 : int~|~\emptyset~|~\emptyset}
    ~~~~
    \inferrule*[]
    {(Y : \alpha_Y) \in \Gamma}
    {\Gamma, \Delta \vdash Y : \alpha_Y~|~\emptyset~|~\emptyset}
    ~~
    \Lambda
}
{\Gamma, \Delta \vdash cons(1,Y): list(\eta)~|~\emptyset~|~\{int = \eta,~\alpha_Y = list(\eta)\} (= T_2) }
\\
\inferrule*[]
{
    \Gamma, \Delta \vdash cons(X,[~]): list(\eta)~|~\emptyset~|~T_1
    ~~~~
    \Gamma, \Delta \vdash cons(1,Y): list(\nu)~|~\emptyset~|~T_2
}
{ \Gamma, \Delta \vdash cons(X,[~]) = cons(1,Y): bool ~|~ C~|~ T_1 \cup T_2 \cup \{list(\nu) = list(\eta)\}    
}
\end{mathpar}
\end{example}

We will now prove that constraint generation is sound, i.e., if we generate constraints any model for them applied to $\Gamma$ and type $\tau$ is derivable in the type system.

\begin{theorem}[ - Soundness of the Constraint Generation]\label{SoundCG}
If $\Gamma,\Delta \vdash t:\tau ~|~C~|~T$ and $\mu \models T$, then $\mu(\Gamma), \Delta \vdash t : \mu(\tau)$ is derivable in the type system.
\end{theorem}
\begin{proof}
We will prove this theorem by induction on the derivation.
\begin{itemize}
    \item If $t$ is a variable $X$, then we have $\Gamma, \Delta \vdash X : \alpha ~|~\emptyset~|~\emptyset$. Any type substitution $\mu$ is such that $\mu \models \emptyset$. And, for any $\mu$, since $\mu(\alpha)$ will be the same in $\Gamma$ and in the consequent of the rule, $\mu(\Gamma), \Delta \vdash X : \mu(\alpha)$ is derivable in the type system by a single application of rule VAR.

    \item If $t$ is a constant $k$, then we have $\Gamma, \Delta \vdash k : \tau[\vec{\alpha} \mapsto \vec{\beta}] ~|~\emptyset~|~\emptyset$, where $(k : \forall \vec{\alpha}. \tau) \in \Delta$. Any type substitution $\mu$ is such that $\mu \models \emptyset$. Then, for any such $\mu$ we can have the derivation in the syntactic system using a single application of rule CST, using $\mu(\tau[\vec{\alpha} \mapsto \vec{\beta}]) = \tau[\vec{\alpha} \mapsto \vec{\mu(\beta)}]$.

    \item If $t$ is a complex term $f(t_1,\dots,t_n)$, then we have $\Gamma, \Delta \vdash f(t_1,\dots,t_n) : \tau[\vec{\alpha} \mapsto \vec{\beta}] ~|~\emptyset~|~T_1 \cup \dots \cup T_n \cup \{\tau\prime_1 \doteq \tau_1[\vec{\alpha} \mapsto \vec{\beta}], \dots, \tau\prime_n \doteq \tau_n[\vec{\alpha} \mapsto \vec{\beta}] \} $, given $\Gamma,\Delta \vdash t_i: \tau_i\prime ~|~ \emptyset ~|~ T_i$, for $i = 1,\dots,n$. We also know that $\mu \models T_1 \cup \dots \cup T_n \cup \{\tau\prime_1 \doteq \tau_1[\vec{\alpha} \mapsto \vec{\beta}], \dots, \tau\prime_n \doteq \tau_n[\vec{\alpha} \mapsto \vec{\beta}] \}$, and any such $\mu$ is such that $\mu \models T_i$ and $\mu \models \tau\prime_i \doteq \tau_i[\vec{\alpha} \mapsto \vec{\beta}]$, for each $i = 1,\dots,n$. By the induction hypothesis, we have $\mu(\Gamma), \Delta \vdash t_i : \mu(\tau\prime_i)$. But we know that $\mu(\tau\prime_i) = \mu(\tau_i[\vec{\alpha} \mapsto \vec{\beta}])$, since $\mu \models \tau\prime_i \doteq \tau_i[\vec{\alpha} \mapsto \vec{\beta}]$, for all $i = 1,\dots,n$. So we also have $\mu(\Gamma), \Delta \Vdash t_i : \mu(\tau_i[\vec{\alpha} \mapsto \vec{\beta}])$. Therefore, by a single application of the CPL rule, we get $\mu(\Gamma), \Delta \Vdash f(t_1,\dots,t_n) : \mu(\tau[\vec{\alpha} \mapsto \vec{\beta}])$.

    \item If $t$ is an equality $t_1 = t_2$, then we have $\Gamma, \Delta \vdash t_1 = t_2 : bool ~|~\{t_1 = t_2\}~|~T_1 \cup T_2 \cup \{\tau_1 \doteq \tau_2\} $, given $\Gamma,\Delta \vdash t_1: \tau_1\prime ~|~ \emptyset ~|~ T_1$ and $\Gamma,\Delta \vdash t_2: \tau_2\prime ~|~ \emptyset ~|~ T_2$. We also know that $\mu \models T_1 \cup T_2 \cup \{\tau_1 \doteq \tau_2\}$, and any such $\mu$ is such that $\mu \models T_1$, $\mu \models T_2$, and $\mu \models \tau_1 \doteq \tau_2$. By the induction hypothesis, we have $\mu(\Gamma), \Delta \vdash t_1 : \mu(\tau_1)$ and $\mu(\Gamma), \Delta \vdash t_2 : \mu(\tau_2)$, but since $\mu \models \tau_1 \doteq \tau_2$, we know that $\mu(\tau_1) \equiv \mu(\tau_2)$. So by a single application of rule EQU, we get $\mu(\Gamma), \Delta \vdash t_1 = t_2 : \mu(bool)$, and $\mu(bool) = bool$.
\end{itemize}
\end{proof}

\subsection{Constraint Solving}

In this section we present a procedure that generalizes Robinson unification \cite{Robinson1965} to account for type constraints and produces solutions, where possible. Since each rule simplifies the constraints, together they induce a straightforward decision procedure for type and term constraints.

Suppose we want to unify two terms $t_1$ and $t_2$. Let us have $\Gamma, \Delta \vdash t_1 = t_2 : bool ~|~C~|~ T$ derived in the constraint generation step. Then we apply the following rewriting rules to the tuple $(C,T)$, until none applies. We apply the rules in order, meaning we only apply rule $n$ if no rule $i$ with $i < n$ applies.

\begin{enumerate}
    \item $(C, \{f(\tau_1,\dots,\tau_n) \doteq f(\tau\prime_1,\dots,\tau\prime_n)\} \cup Rest) \to (C, \{\tau_1 \doteq \tau\prime_1,\dots,\tau_n \doteq \tau\prime_n\} \cup Rest)$
    \item $(C, \{\tau \doteq \tau\} \cup Rest) \to (C, Rest)$
    \item $(C, \{f(\tau_1,\dots,\tau_n) \doteq g(\tau\prime_1,\dots, \tau\prime_m)\}\cup Rest) \to wrong$, if $f \neq g$ or $n \neq m$
    \item $(C, \{\tau\doteq\alpha\}\cup Rest) \to (C, \{\alpha\doteq\tau\}\cup Rest)$, if $\tau$ is not a type variable
    \item $(C, \{\alpha\doteq\tau\}\cup Rest) \to (C, \{\alpha\doteq\tau\}\cup Rest[\alpha \mapsto \tau])$, if $\alpha$ does not occur in $\tau$
    \item $(C, \{\alpha\doteq\tau\}\cup Rest) \to wrong$, if $\alpha$ occurs in $\tau$

    \item $(\{f(t_1,\dots,t_n) = f(s_1,\dots,s_n)\} \cup Rest,T) \to (\{t_1 = s_1,\dots,t_n = s_n\} \cup Rest,T)$
    \item $( \{t = t\} \cup Rest,T) \to (Rest,T)$
    \item $(\{f(t_1,\dots,t_n) = g(s_1,\dots,s_m)\}\cup Rest,T) \to false$, if $f\neq g$ or $n \neq m$
    \item $(\{t = X\}\cup Rest,T) \to (\{X = t \}\cup Rest,T)$, if $t$ is not a variable
    \item $(\{X = t\}\cup Rest,T) \to (\{X = t\}\cup Rest[X \mapsto t],T)$, if $X$ does not occur in $t$
    \item $(\{X = t\}\cup Rest,T) \to false$, if $X$ occurs in $t$.
\end{enumerate}

\noindent
We will use the symbol $\to^{*}$ to denote the reflexive and transitive closure of $\to$.

\begin{example}
Let  $C = \{cons(X,[~]) = cons(1,Y)\}$ and $T = \{\alpha_X = \nu,~list(\gamma) = list(\nu), int = \eta,~\alpha_Y = list(\eta), list(\nu) = list(\eta)\}$. Step-by-step the algorithm rewrite the pair $(C,T)$ as follows:\\
$(C,T) \to (C,\{\alpha_X = \nu,~\gamma = \nu, int = \eta,~\alpha_Y = list(\eta), list(\nu) = list(\eta)\}) \to$\\
$(C,\{\alpha_X = \nu,~\gamma = \nu, int = \eta,~\alpha_Y = list(\eta), \nu = \eta\}) \to$\\
$(C,\{\alpha_X = \nu,~\gamma = \nu, \eta = int,~\alpha_Y = list(\eta), \nu = \eta\}) \to$\\
$(C,\{\alpha_X = \nu,~\gamma = \nu, \eta = int,~\alpha_Y = list(int), \nu = int\}) \to$\\
$(C,\{\alpha_X = \nu,~\gamma = int, \eta = int,~\alpha_Y = list(int), \nu = int\} ( = T\prime)) \to$\\
$(\{X = 1,[~] = Y\},T\prime) \to (\{X = 1,Y = [~]\},T\prime)$\\

Note that, in the final pair, no more rules apply and we can interpret this pair as a pair of substitutions for terms and for types, respectively.
\end{example}

\subsection{Properties of the Regular Typed Unification Algorithm}

In this section we show the main properties of regular typed unification.
Firstly, it always terminates. Secondly, it is correct, meaning that the result is the same as we would have gotten in the equality theory defined for $=$ semantically. One big obstacle for this second property is that terms may not be ground when we want to unify them, and semantically we always need a state to evaluate variables. We will be conservative and assume that if there is a possible state for which the terms have values in the same semantic domain, then there is no type error (yet). Similarly, if there is a state for which the terms have the same semantic value, then the result is not \textit{false} (yet).


\begin{theorem}[ - Termination]
Let $(C,T)$ be the sets of constraints generated for terms $t_1$ and $t_2$. The algorithm always terminates, returning a pair of unifiers, false, or wrong.
\end{theorem}
\begin{proof}
We divide the algorithm in two parts. The first consists of the rules 1 to 6, and the second of the rules 7 to 12.
Each of these parts are the Martelli-Montanari algorithm \cite{MartelliMontanari} for its corresponding kind of constraints, type equality and equality, respectively. Therefore they terminate.

For a formal proof for the termination of the Martelli-Montanari algorithm, we defer the reader to \cite{MartelliMontanari}.

Moreover, if the Martelli-Montanari terminates, the output is either a most general unifier, or the algorithm fails. In the first part, failure is represented by \textit{wrong}, and in the second part, it is represented by \textit{false}. So our algorithm either terminates and outputs \textit{wrong}, \textit{false}, or both parts succeed and the algorithm outputs a pair of most general unifiers.
\end{proof}\\

We now know that the algorithm terminates, and what the outputs might be. We will additionally prove that the result is semantically valid. We start by proving a few auxiliary lemmas.

The following lemmas are used to prove soundness.

\begin{lemma}[ - Rewriting Consistency]\label{consistent}
Let $(C,T) \to (C\prime,T\prime)$ be a step in the typed unification algorithm, such that the output is not $false$ nor $wrong$. Then, if for all equality constraints $(t_1 = t_2) \in C\prime$ the substitution $\theta$ is a unifier of $t_1$ and $t_2$, then $\theta$ is also a unifier of each equality constraint in $C$. Same applies to $T\prime$ and $T$, with a type substitution $\mu$.
\end{lemma}
\begin{proof}
We will prove this by case analysis.
\begin{enumerate}
    \item $C\prime$ and $C$ are equal so, trivially, any substitution $\theta$ that unifies each equality constraint in $C$ also unifies each equality constraint in $C\prime$. Now suppose that $\mu$ is a type substitution such that $\mu(\tau_i) = \mu(\tau_i\prime)$, for $i = 1,\dots,n$, then, also $\mu(f(\tau_1,\dots,\tau_n)) = \mu(f(\tau\prime_1, \dots, \tau\prime_n))$. All other type equality constraints in $T$ are also in $T\prime$, so any unifier of $T\prime$ is a unifier of $T$.

    \item $C\prime$ and $C$ are equal so, trivially, any substitution $\theta$ that unifies each equality constraint in $C$ also unifies each equality constraint in $C\prime$. All type equality constraints in $T\prime$ are also in $T$, so all unifiers of $T\prime$ are unifiers of that subset of $T$. Moreover, $T$ has one more type equality constraint $\tau\doteq \tau$, but any substitution, in particular any unifier of $T\prime$ is also a unifier of $\tau$ with itself.

    \item This case does not apply, since the output is $wrong$.

    \item $C\prime$ and $C$ are equal so, trivially, any substitution $\theta$ that unifies each equality constraint in $C$ also unifies each equality constraint in $C\prime$. Any unifier of $T\prime$ is also a unifier of $T$, since swapping the terms on a type equality constraints does not change the fact that a substitution is a unifier.

    \item $C\prime$ and $C$ are equal so, trivially, any substitution $\theta$ that unifies each equality constraint in $C$ also unifies each equality constraint in $C\prime$. Suppose $\mu$ is a unifier of $T\prime$, then $\mu(\alpha) = \mu(\tau)$. Therefore, since $T\prime = T[\alpha \mapsto \tau]$ for all constraints except $\alpha \doteq \tau$, then $\mu(T\prime) = \mu(T[\alpha \mapsto \tau]) = (\mu \circ [\alpha \mapsto \tau])(T)$ but since $\mu(\alpha) = \mu(\tau)$, then $(\mu \circ [\alpha \mapsto \tau])(T) = \mu(T)$. So $\mu$ is also a unifier of $T$.

    \item This case does not apply, since the output is $wrong$.
\end{enumerate}
The proof for the rest of the cases is similar to the proof for the cases 1 to 6, except we replace type equality constraints with equality constraints, type substitution with substitution, and $wrong$ with $false$.
\end{proof}

\begin{lemma}[ - Self-satisfiability]\label{SelfSat}
Suppose $C$ is a set of equality constraints in normal form. Then, $C$ can be interpreted as a substitution $\theta$, and $\theta$ is a unifier of all constraints in $C$. Same can be said for a set of type equality constraints in normal form $T$.
\end{lemma}

\begin{proof}
If $C$ is in normal form, then $C = \{X_1 = t_1,\dots, X_n = t_n\}$, where $X_i$ is a variable and none of $X_i$ occurs in any $t_i$. So, when we interpret $C$ as a substitution $\theta$, we will have $\theta = [X_1 \mapsto t_1,\dots,X_n \mapsto t_n]$. When we apply $\theta$ to each constraint in $C$, we will get $\theta(C) = \{\theta(X_1) = \theta(t_1), \dots, \theta(X_n) = \theta(t_n)\}$, but since none of the variables $X_i$ occur in any $t_i$, then $\theta(t_i) = t_i$. Moreover, $\theta(X_i) = t_i$. So we get $\theta(C) = \{ t_1 = t_1, \dots t_n = t_n\}$. Therefore, $\theta$ is a unifier of all constraints in $C$. The proof for type equality constraints is similar to this one, replacing substitutions with type substitutions and terms with type terms.
\end{proof}\\

We are now ready to prove the following theorem that proves that the algorithm outputs a semantically correct value.

\begin{theorem}[ - Soundness of the Typed Unification Algorithm]
Let $t_1$ and $t_2$ be the input to the typed unification algorithm, and $\Gamma \vdash t_1 = t_2 ~|~C~|~T$. Suppose $(C,T) \to^{*} R$.
\begin{enumerate}
\item If $R = (\theta,\mu)$, a pair of substitutions for terms and types respectively, then $\theta, \mu \models C,T$.

\item If $R = false$, then there is no substitution $\theta$ such that $\theta \models C$, but there is a type substitution $\mu$ such that $\mu \models T$.

\item If $R = wrong$, then there is no type substitution $\mu$ such that $\mu \models T$.
\end{enumerate}
\end{theorem}
\begin{proof}
The proof for (1) follows from Lemmas \ref{consistent} and \ref{SelfSat}. We get that $\theta$ and $\mu$ are unifiers of $C$ and $T$, respectively.

The proof for (2) follows from the fact that the Martelli-Montanari algorithm is complete, i.e., if there was a unifier for the equality constraint set $C$, then it would have been obtained. Therefore there is no unifier for $C$, so there is no $\theta$ such that $\theta \models C$. However, since we got to the second part of the algorithm, then we were able to find a unifier for the type equality constraints. This means that there is at least one unifier for $T$, so there is a $\mu$ such that $\mu \models T$.

The proof for (3) follows form the fact that the Martelli-Montanari algorithm is complete, i.e., if there was a unifier for the type equality constraint set $T$, then it would have been obtained. Therefore there is no unifier for $T$, so there is no $\mu$ such that $\mu \models T$.
\end{proof}\\

We also prove that the unification algorithm outputs principal typings when it succeeds.
\begin{theorem}[ - Completeness of the Typed Unification Algorithm]
Let $t$ be term, or a unification of two terms, $\Gamma$ be a generic context, and $\Delta$ be type assumptions for constants and function symbols. If $\Gamma, \Delta \vdash t : \tau~|~C~|~T$ and $(C,T) \to^{*} (\theta,\mu)$. Then $(\mu(\Gamma),\mu(\tau))$ is a principal typing of $t$.
\end{theorem}

\begin{proof}
We will prove by structural induction on $t$.
\begin{itemize}
    \item If $t$ is a variable $X$, then $(X : \alpha) \in \Gamma$. We know that $\Gamma, \Delta \vdash X : \alpha~|~\emptyset~|~\emptyset$ by a single application of GVAR. $(\emptyset,\emptyset) \to ([~],[~])$, where $[~]$ are each the identity substitution for variables and type variables, respectively. Therefore $[~](\alpha) = \alpha$, and any type $\tau$ derived in the type system such that $(X:\tau) \in \Gamma\prime$, then $\tau$ is an instance of $\alpha$.

    \item If $t$ is a constant $k$, then $(k : \forall \vec{\alpha}. \tau) \in \Delta$. We know that $\Gamma, \Delta \vdash k : \tau[\vec{\alpha} \mapsto \vec{\sigma}]$, for any $\vec{\sigma}$. We get by a single application of rule GCST that $\Gamma, \Delta \vdash k : \tau[\vec{\alpha} \mapsto \vec{\beta}]~|~\emptyset~|~\emptyset$. $(\emptyset,\emptyset) \to ([~],[~])$, where $[~]$ are each the identity substitution for variables and type variables, respectively. Therefore $[~](\tau[\vec{\alpha} \mapsto \vec{\beta}]) = \tau[\vec{\alpha} \mapsto \vec{\beta}]$, and we known that any $\vec{\sigma}$ is an instance of $\vec{\beta}$.

    \item If $t$ is a complex term $f(t_1,\dots,t_n)$, then $(f : \forall \vec{\alpha}. \tau_1 \times \dots \times \tau_n \to \tau\prime) \in \Delta$. By the induction hypothesis, we know that if $(\emptyset,T_i) \to^{*} ([~],\mu_i)$, then $(\mu_i(\Gamma),\mu_i(\tau\prime_i))$ is a principal typing of $t_i$.
Now suppose $(\emptyset,T_1\cup \dots \cup T_n \cup \{\tau\prime_1 \doteq \tau_1[\vec{\alpha} \mapsto \vec{\beta}], \dots, \tau\prime_n \doteq \tau_n[\vec{\alpha} \mapsto \vec{\beta}] \}) \to^{*} ([~],\mu)$. We know that $\mu(\tau\prime_i) = \mu(\tau_i[\vec{\alpha} \mapsto \vec{\beta}])$ for all $i = 1,\dots,n$. So we can derive $\mu(\Gamma),\Delta \vdash t_i : \mu(\tau_i[\vec{\alpha} \mapsto \vec{\beta}])$, and by a single application of rule CPL, we get $\mu(\Gamma),\Delta \vdash f(t_1,\dots,t_n) : \mu(\tau[\vec{\alpha} \mapsto \vec{\beta}])$.
Now we need to prove that this typing $(\mu(\Gamma),\mu(\tau[\vec{\alpha} \mapsto \vec{\beta}]))$ is the principal typing. Suppose we had another typing that was not an instance of this one $(\mu\prime(\Gamma),\mu\prime(\tau[\vec{\alpha} \mapsto \vec{\beta}]))$. Since $\mu$ is an MGU of $T_1\cup \dots \cup T_n \cup \{\tau\prime_1 \doteq \tau_1[\vec{\alpha} \mapsto \vec{\beta}], \dots, \tau\prime_n \doteq \tau_n[\vec{\alpha} \mapsto \vec{\beta}] \}$, then either for some $i$ $\mu\prime \nvDash T_i$, or for some $i$ $\mu\prime \nvDash \tau\prime_i \doteq \tau_i[\vec{\alpha} \mapsto \vec{\beta}]$. If the former is true, then $(\mu\prime(\Gamma),\mu(\prime(\tau\prime_i)))$ is not an instance of the principal typing for $t_i$ and therefore cannot be derived in the type system. If the latter is true, then $\mu\prime(\tau_i[\vec{\alpha} \mapsto \vec{\beta}]) \neq \mu\prime(\tau\prime_i)$ and we cannot use the rule CPL in the type system. Therefore, $(\mu(\Gamma),\mu(\tau[\vec{\alpha} \mapsto \vec{\beta}]))$ is the principal typing for $f(t_1,\dots,t_n)$.

    \item Suppose $t$ is an equality $t_1 = t_2$. By the induction hypothesis, we know that if $(\emptyset,T_i) \to^{*} ([~],\mu_i)$, then $(\mu_i(\Gamma),\mu_i(\tau_i))$ is a principal typing of $t_i$. Now suppose $(\{t_1 = t_2\},T_1\cup T_2\cup\{\tau_1 \doteq \tau_2) \to^{*} (\theta,\mu)$. We know that $\mu(\tau_i)$ is an instance of $\mu_i(\tau_i)$, so we can derive $\mu(\Gamma),\Delta \vdash t_i : \mu(\tau_i)$ in the type system. By a single application of rule EQU, we get $\mu(\Gamma),\Delta \vdash t_1 = t_2 : bool$. So $(\mu(\Gamma),\mu(bool))$ is a typing. Suppose it was not the principal typing. Suppose we had another typing that was not an instance of this one $(\mu\prime(\Gamma),\mu\prime(bool))$. For all $\mu$, $\mu(bool) = bool$. Since $\mu$ is an MGU of $T_1 \cup T_2 \cup \{\tau_1 \doteq \tau_2\}$, so if $\mu\prime$ is not an instance, then either for some $i$ $\mu\prime \nvDash T_i$ or $\mu\prime \nvDash \tau_1 \doteq \tau_2$. If the former is true, then $(\mu\prime(\Gamma),\mu\prime(\tau_1))$ is not an instance of its principal typing, so it cannot be derived in the type system. If the latter is true, the types for $t_1$ and $t_2$ are different and we cannot apply rule EQU, so we could not have this derivation in the type system. Therefore, $(\mu(\Gamma),\mu(bool))$ is the principal typing for $t_1 = t_2$.
\end{itemize}
\end{proof}

\section{Final Remarks} 
Our regular typed unification algorithm provides some foundation for the use of regular types to dynamically catch erroneous Prolog behaviors. Indeed, one of the original motivations for this work was to understand how to extend the YAP Prolog system \cite{costa2012yap} with an effective dynamic typing.
In \cite{BarbosaFloridoCosta22} we proposed a typed SLD-resolution (TSLD) which used our previous notion of typed unification. Our goal now is to effectively extend SLD resolution with unification of terms typed by deterministic regular types.
A TSLD-tree branch may result in \textit{true}, \textit{false}, or \textit{wrong}, depending on the same results for the unifications in the branch. In \cite{BarbosaFloridoCosta22}, each TSLD-tree branch where a unification outputs \textit{false}, needed to continue execution on the same branch in order to check if there was a type error in some other atom in the query.
This leads to a drastic increase in the runtime of programs.
\begin{example}\label{ex7}
Consider the following (unrealistic) but possible program:
\begin{verbatim}
p(0).
\end{verbatim}
and query: \texttt{?- p(1),\dots,p(900),p(a)}.
In Prolog SLD-resolution the query fails after one SLD-step. 
In the TSLD-resolution defined in \cite{BarbosaFloridoCosta22}, since the first 900 queries return $false$, one needs to reach step 901 in order to obtain the value $wrong$.
\end{example}

We argue that when adding regular typed unification to Prolog we must have a compromise between completeness and efficiency. 
If the evaluation of a query is $false$ we stop execution, and the same happens for $wrong$. However, if the result is $false$ and there are other atoms in the query, we cannot assure that the value for that branch is indeed $false$, only that it is not $true$. Thus, in our extension to Prolog we output \textbf{no}(?) in these cases. On the other hand, we output \textbf{no}(false) if there are no other atoms in the query, and  \textbf{no}(wrong) if the branch ends on $wrong$.
Note that in many programs, for some queries, we are \textit{always} able to detect the type error.
\begin{example}
Consider the predicate that calculates the length of a list:
\begin{verbatim}
length([], 0).
length([_|T],N) :- length(T,N1), N is N + 1.
\end{verbatim}
One typical bug is to swap the arguments of a predicate.
Now note that, in this case, if we have the erroneous query \texttt{?- length(3,[a,b,c])}, both branches of the TSLD-tree output \textit{wrong} since there is a type error in the first argument (and also in the second).
\end{example}

\section{Related Work}

This paper generalizes a typed unification algorithm previously defined by the authors in \cite{BarbosaFloridoCosta22} that was used in the dynamic typing of logic programs. In \cite{BarbosaFloridoCosta22}, functions symbols $f$ of arity $n$ had co-domains which were always sets of terms of the form $f(t_1,\dots,t_n)$, where the arguments $t_i$ belong to the corresponding domain of $f$. Here we extend this notion enabling the use of semantic domains and co-domains described by deterministic regular types allowing a non-Herbrand interpretation of function symbols.

The most obvious related work is many-sorted unification \cite{Walther1988}, though many-sorted unification assumes an infinite hierarchy of sorts and we do not assume a hierarchy of types. In particular there is a relation with many-sorted unification with a forest-structured sort hierarchy \cite{Walther1988}, but even compared with this strong restricted unification problem, our work gives easier and nicer results, mostly due to the use of an expressive universe partition based on regular types but with no underlying hierarchy on the domains. 

Here we study unification of terms interpreted in domains described by regular types, and we allow a form of parametric polymorphism in the description of term variables. Parametric polymorphic descriptions of sorted domains goes back to Smolka generalized order-sorted logic \cite{Smolka1988}. In his system, subsort declarations are propagated to complex type expressions, thus the main focus is on subtyping which is not the scope of our work.

Dart and Zobel \cite{DBLP:books/mit/pfenning92/DartZ92} provided an algorithm for regular type unification, generating a type unifier. Due to problems related to tuple distributivity, not all types had a most general type unifier. In consequence, unification returned a weak type unifier. However, the question whether unification returned a minimal weak type unifier was unknown and left as an open question.

In \cite{HANUS199163}, there is a typed unification algorithm used in a typed operational semantics for logic programming. The main difference to our work is that in  \cite{HANUS199163} failing unification due to ill-typedness is not detected with a different value and it is not different from a well-typed failed unification.

A data type reconstruction algorithm was previously defined in \cite{SchrijversB06} based on equations and inequations constraints. This was also applied to logic programs (terms and predicates). Here we focus on term unification, thus equality is the only predicate, and this rather simplifies our type system. \\

{\bf Acknowledgements} This work was partially financially supported by UIDB/00027/2020 of the Artificial Intelligence and Computer Science Laboratory, LIACC, funded by national funds through the FCT/MCTES (PIDDAC).

\bibliography{bibliography}

\end{document}